\documentclass[aps,pra,twocolumn,floatfix,superscriptaddress]{revtex4-2}
\usepackage{float}
\usepackage{amsmath}
\usepackage{amssymb}
\usepackage{graphicx}
\usepackage{mathrsfs}
\usepackage{hyperref}
\usepackage{lipsum}
\usepackage[procnames]{listings}
\usepackage{layouts}
\usepackage{epstopdf}
\usepackage{dsfont}
\usepackage[normalem]{ulem}
\usepackage{tabularx}

\usepackage{amsmath,amssymb,amsthm}
\usepackage{algorithm}
\usepackage{algpseudocode}
\usepackage{float}
\usepackage{xcolor}

\usepackage[dvipsnames]{xcolor} 
\usepackage{ragged2e}           
\usepackage{tikz}
\usetikzlibrary{
    quantikz2,
    decorations.pathmorphing
}

\newcolumntype{M}{>{\displaystyle}r@{\hspace{1pt}\pm\hspace{1pt}}l}

\begin{document}

\title{Fault-Tolerant Logical Operations and Efficient State Preparation \\ in Modular Quantum Architectures with Noisy Interfaces}

\author{S. Siddardha Chelluri}
\author{Riccardo Mengoni}
\affiliation{Welinq, 14 rue Jean Macé, 75011 Paris, France}
\author{Tom Darras}
\affiliation{Welinq, 14 rue Jean Macé, 75011 Paris, France}
\author{Julien Laurat}
\affiliation{Laboratoire Kastler Brossel, Sorbonne Université, CNRS,
ENS-Universite PSL, Collège de France, 4 Place Jussieu, 75005 Paris, France}
\author{Eleni Diamanti}
\affiliation{LIP6, Sorbonne Université, CNRS, 4 Place Jussieu, 75005 Paris, France}
\author{Ioannis Lavdas}
\email[email: ]{ioannis.lavdas@welinq.fr}
\affiliation{Welinq, 14 rue Jean Macé, 75011 Paris, France}

\begin{abstract}
Modular quantum computing is a leading paradigm for scaling quantum computation beyond the resource limitations of monolithic devices. In this architecture,  multiple quantum processing units (QPUs), employing identical or distinct qubit modalities, are interconnected via shared entanglement.  Here, we investigate how errors at module interfaces and within individual QPUs affect fault-tolerant computation when qubits are encoded using the rotated surface code. Going beyond the logical-memory benchmark, we perform circuit-level simulations of fault-tolerant nonlocal CNOT gates implemented via lattice surgery between QPUs connected by noisy Bell pairs, and analyze the resulting logical error rates. Our results show that interfaces can tolerate noise up to an order of magnitude higher than intra-QPU noise, with only a minor reduction in the fault-tolerance threshold. Using the CNOT as a building block, we further develop an efficient protocol for preparing distributed fault-tolerant logical GHZ states, reducing ancillary overhead, time, and nonlocal Bell-pair consumption. We show that ancilla minimization in this setting is equivalent to a vertex-cover problem on an associated graph,
and introduce a polynomial-time heuristic algorithm for finding low-overhead solutions. These results quantify the interface-noise tolerance and resource trade-offs of distributed logical operations, informing the design of fault-tolerant quantum computers built from interconnected processors.
    
\end{abstract}

\maketitle

\section{Introduction}
Distributed quantum computing offers a scalable route beyond monolithic devices by distributing  algorithms across multiple quantum processing units (QPUs) interconnected via entanglement and classical communication~\cite{Caleffi_2024,Barral:2024cef,Diamanti:26}.
Recent advances in atomic memories, photonic interconnects, and networked processors have brought modular architectures closer to realization ~\cite{cirac1999dqc,monroe2014modular,nickerson2014scalable,sunami2025scalable, bourassa2021photonicft, awschalom2021quics, covey2023neutralnetwork, aghaee2025networking}. In these architectures, nonlocal quantum operations constitute the fundamental primitive for distributed quantum computation, enabling interactions between qubits hosted on different QPUs. Nonlocal operations rely on shared entanglement and teleportation-based protocols, making Bell-pair consumption, classical communication rounds, and communication latency key resources~\cite{gottesman1999teleportation,eisert2000nonlocal}.

These constraints have motivated extensive research on distributed circuit compilation and partitioning, including hypergraph-based methods and gate-reordering techniques that minimize the overhead associated with inter-QPU communication~\cite{andresmartinez2019hypergraph,mengoni2026reordering}. Such approaches, however, are generally studied at the level of unencoded circuits. Their resource advantages need to be reassessed when computations are implemented with logical qubits, where error correction introduces additional requirements for physical qubits, entanglement, and execution time. Protocols that are efficient at the unencoded-circuit level need not remain efficient when implemented fault-tolerantly using logical qubits. 

Previous studies of distributed quantum error correction have largely assessed resilience to noisy interconnects through logical-memory benchmarks, which characterize the ability to preserve encoded quantum information under repeated error correction rounds~\cite{de_Bone_2024, Ramette2024, sinclair2024faulttolerantopticalinterconnectsneutralatom,jacinto2026network, sutcliffe2025distributed, moylett2026logicalgatesfloquetcodes}. Distributed computation additionally requires logical operations between QPUs, whose performance must be assessed over the complete gate protocol. Lattice surgery provides a framework for implementing such operations in surface-code architectures through sequences of joint stabilizer measurements \cite{horsman2012latticesurgery, Litinski_2018,Litinski_2019, Chamberland_2022}. Although now well established for monolithic processors \cite{KITAEV20032,horsman2012latticesurgery,litinski2019surfacecodes, qec_google, lin2026surfacecodelogicaloperations}, its implementation across noisy inter-QPU links raises additional questions about the reliability and resource requirements of complete fault-tolerant logical-gate protocols \cite{haug2025latticesurgerybellmeasurements}.

The cost of a logical operation in distributed framework is determined not only by the underlying error-correcting code but also by the allocation of logical and ancilla patches across QPUs and by the scheduling of lattice surgery operations. Additional ancilla patches increase the physical-qubit footprint within each module, whereas longer operation sequences increase execution time and exposure to memory errors. Consequently, distributed lattice surgery naturally gives rise to a resource-allocation problem in which different layouts can realize the same logical operation while exhibiting different trade-offs between spatial overhead and temporal cost. Despite their importance for the design of modular fault-tolerant quantum architectures, these trade-offs remain largely unexplored at the circuit level in the context of logical operations across multiple QPUs. In this work, we systematically investigate this regime.


\begin{figure}[t!] 
\centering
\hspace*{-0.1cm}%
\includegraphics[width=1.3\linewidth]{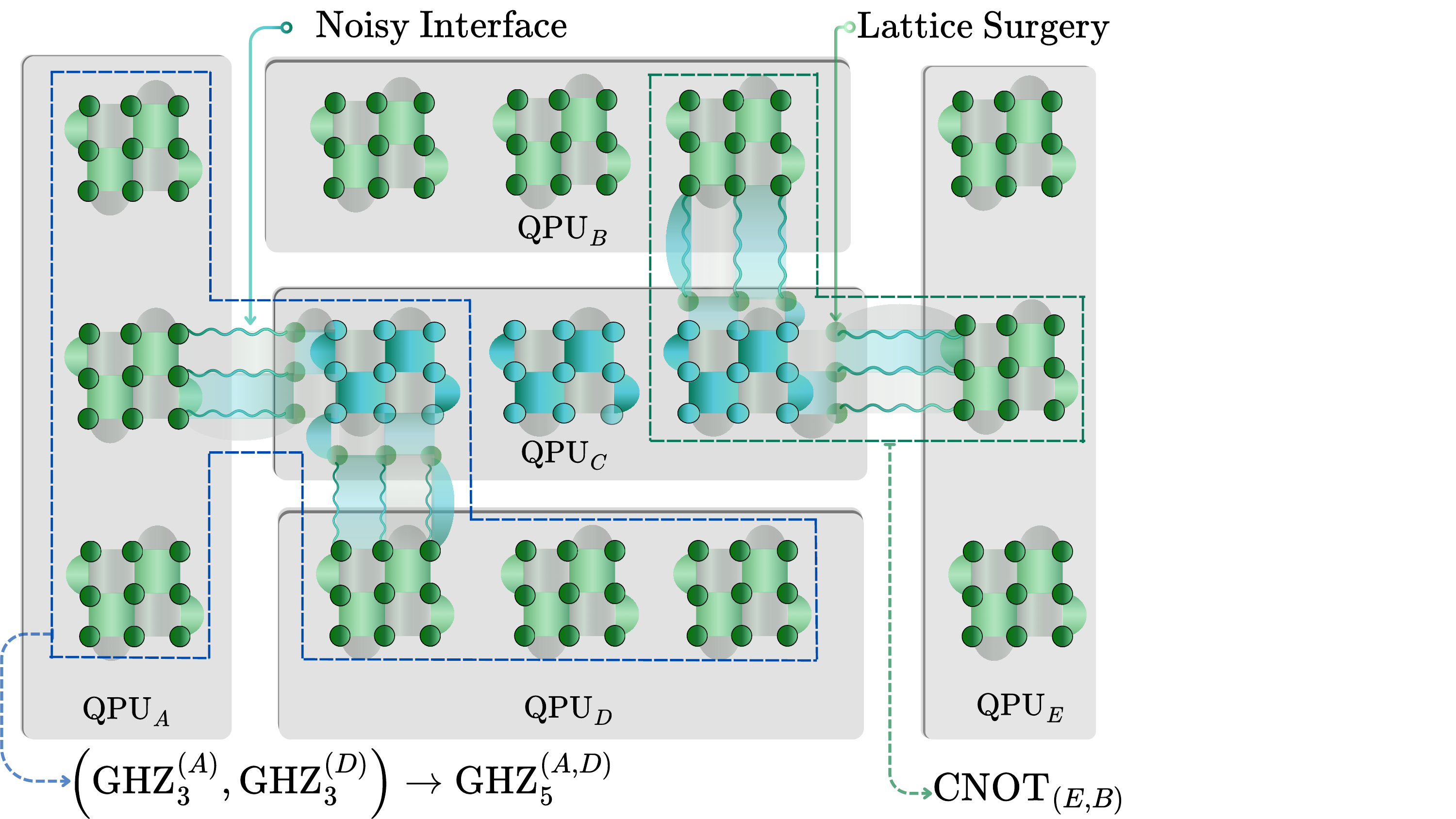} 
\caption{{Modular architecture with rotated surface-code encoding. 
Green-grey patches represent logical data qubits, while blue-grey patches logical ancilla qubits. Remote fault-tolerant CNOT gates are implemented by lattice surgery across noisy interfaces (wiggly lines): on the right part we show a CNOT controlled by a qubit in $\text{QPU}_{E}$, targeting one in $\text{QPU}_{B}$ via an ancilla qubit in $\text{QPU}_{C}$. Such CNOTs can then be used to fuse multiple locally prepared GHZ states into larger distributed GHZ states, enabling the creation of multipartite entanglement spanning several QPUs. This is seen in the left part of the figure, where a GHZ in $\text{QPU}_{A}$ is fused with one in $\text{QPU}_{D}$ via an ancilla in $\text{QPU}_{C}$.}}
\label{fig:Fig1}
\end{figure}

Building on prior studies of distributed quantum error correction \cite{Ramette2024, sinclair2024faulttolerantopticalinterconnectsneutralatom, Nickerson2013}, our work extends threshold analysis from preserving logical information to executing fault-tolerant logical operations between QPUs. We study nonlocal logical CNOT gates in rotated surface-code architectures \cite{PhysRevA.76.012305, PhysRevLett.108.180501} through circuit-level simulations of the full lattice-surgery protocol, incorporating both local QPU and interconnect noise and decoding the resulting syndromes to estimate logical failure probabilities (see Fig.~\ref{fig:Fig1}). Although memory experiments and CNOT protocols share physical error-correction primitives, memory benchmarks do not by themselves quantify the reliability of logical operations. Our simulations explicitly assess this reliability and quantify the contribution of noisy interconnects to the logical error budget, providing a benchmark directly relevant to distributed quantum computation.


Having characterized the nonlocal logical CNOT primitive, we now use it to prepare distributed logical GHZ states and analyze their resource requirements. GHZ states are relevant to quantum sensing and networked quantum technologies \cite{Hyllus_2012,Kielinski_2024,Main_2025} and provide benchmarks for quantum processors \cite{moses2023racetrack,bao2024globalghz,cruz2019ghzw,chen2023ghz127,javadiabhari2025bigcats,chelluri2025}. Prior approaches include generation across quantum networks with one GHZ qubit per node \cite{MMG2019,chelluri2026,Sen2025} and within monolithic processors \cite{chelluri2025, moses2023racetrack}. Here, we study the fault-tolerant entangled state preparation across logical patches distributed over QPUs connected by noisy interconnects as shown in Fig.~\ref{fig:Fig1}. We characterize trade-offs between ancilla overhead, execution time, and inter-QPU Bell pair consumption across network topologies, and formulate ancilla allocation as a vertex-cover problem over spanning trees, and introduce a polynomial-time heuristic for finding low-overhead solutions.  Together, these contributions provide a framework for assessing the architectural cost of distributed fault-tolerant operations and state engineering.




\section{Distributed Fault-Tolerant CNOT Gate}  

We consider a lattice-surgery implementation of the logical CNOT between QPUs. Logical qubits are encoded in rotated surface-code patches.  In the simplest configuration, the control patch is hosted in one QPU, while the ancilla and target patches reside in a second QPU. The protocol follows the standard lattice-surgery CNOT construction described in Appendix~A, with the difference that the joint operation between control and ancilla patches must be performed across the interface. Consequently, only stabilizer checks that cross the inter-QPU boundary require a nonlocal implementation, while all stabilizer measurements fully contained within a single QPU follows the local lattice-surgery construction.

The required cross-QPU CNOT gates required across the interface are implemented via gate teleportation using noisy Bell pairs \cite{1635966}. Compared with the local protocol, this nonlocal implementation introduces additional errors from Bell-pair imperfections and from the extra gates involved in the gate-teleportation protocol (see Appendix~C). To incorporate these effects, we use a circuit-level noise model in which each single qubit and two-qubit  gate is followed by the corresponding depolarizing error channel. Local two-qubit gates are assigned error probability $p$, whereas nonlocal two-qubit gates are assigned error probability $\alpha p$, where $\alpha$ is the \textit{interface noise parameter} that captures the additional noise associated with Bell-pair generation and gate-teleportation. Explicitly, a two-qubit gate $U$ is followed by:
\begin{equation}
\rho \mapsto (1-p_g)\,U\rho U^\dagger
+ \frac{p_g}{15}\sum_{P \in \mathcal{P}_2 \setminus \{II\}}
P(U\rho U^\dagger)P^\dagger ,
\end{equation}
where $\mathcal{P}_2=\{I,X,Y,Z\}^{\otimes 2}$, the sum runs over the 15 nontrivial two-qubit Pauli operators $P$, each applied with equal probability, and finally $p_g$ denotes the gate error probability: $p_g= \epsilon$ for local two-qubit gates and $p_g=\alpha \cdot \epsilon$ for nonlocal two-qubit gates. 
We use this nonlocal CNOT implementation as a building block to study fault-tolerance thresholds beyond 
memory experiments \cite{sinclair2024faulttolerantopticalinterconnectsneutralatom}.


\section{Simulation Results of Distributed CNOT}
We estimate logical failure probabilities by simulating complete lattice-surgery circuits under the noise model described above. Three distributed configurations are considered: $(a)$ a bipartite {configuration}, where the control qubit is located in QPU$_A$, while  ancilla and target qubits are in QPU$_B$; $(b)$ a bipartite {configuration} where QPU$_B$ includes ancilla and two targets and $(c)$ a three-partite configuration, where  control, ancilla and target qubits are located in separate QPUs.

For each configuration, we compare the distributed configuration with the corresponding monolithic one. In the latter, all two-qubit gates are assumed to have the same fixed error rate, obtained with $\alpha=1$. In the distributed implementation instead, two-qubit gates acting on qubits located in different QPUs are assigned a higher error rate, associated with $\alpha=20$, in order to account for the additional noise introduced by gate teleportation and Bell-pair noise, as explained in the Appendix~C. We assume identical noise characteristics for control, ancilla and target as well as for data and syndrome qubits, reflecting the assumption that all physical qubits are equivalent. 

Each simulation run consists of qubit initialization using rotated surface-code patches of a chosen code distance, followed by lattice surgery implementing the fault-tolerant CNOT and concluding with final logical qubit measurements. We use 50,000 runs for each simulated parameter set.
Circuits are simulated using LOOM \cite{el_loom_github} and STIM \cite{gidney2021stim}, while syndrome decoding is performed using PyMatching \cite{Higgott2025sparseblossom}. In case (c), where two CNOTs are applied, even a single CNOT failure is considered as logical failure. The noise model implemented is the single and two qubit depolarizing noise applied after each single and two-qubit gate, respectively.

The results of the simulations are shown in Fig.~\ref{fig:three-column-composite}. Each figure includes the distributed architecture used to implement the fault-tolerant operation, together with the corresponding threshold plot showing how logical error rate $\varepsilon_{\text{logical}}$ varies as a function of  physical error rate $\varepsilon_{\text{phys}}$. Additionally, each instance includes the threshold plot of the corresponding monolithic case. 
The distributed implementations exhibit a minor reduction in the fault-tolerance threshold relative to the corresponding monolithic implementations, with shifts ranging from $2\cdot10^{-3}$ to $3\cdot10^{-3}$. The largest reduction occurs for the two-CNOT implementation using two interfaces.

Therefore, the weak dependence on the inter-QPU gate error rate indicates that the threshold behavior is dominated by local gate noise. The key quantity is not the absolute number of nonlocal two-qubit gates, but rather its ratio with respect  to local two-qubit gates, which remains small for the architectures considered. A corresponding numerical analysis and bounds on this ratio are provided in Appendix~B. 

\begin{figure}[!t]
    \centering
        \hspace{-3mm}
        \includegraphics[width=1.00\linewidth]{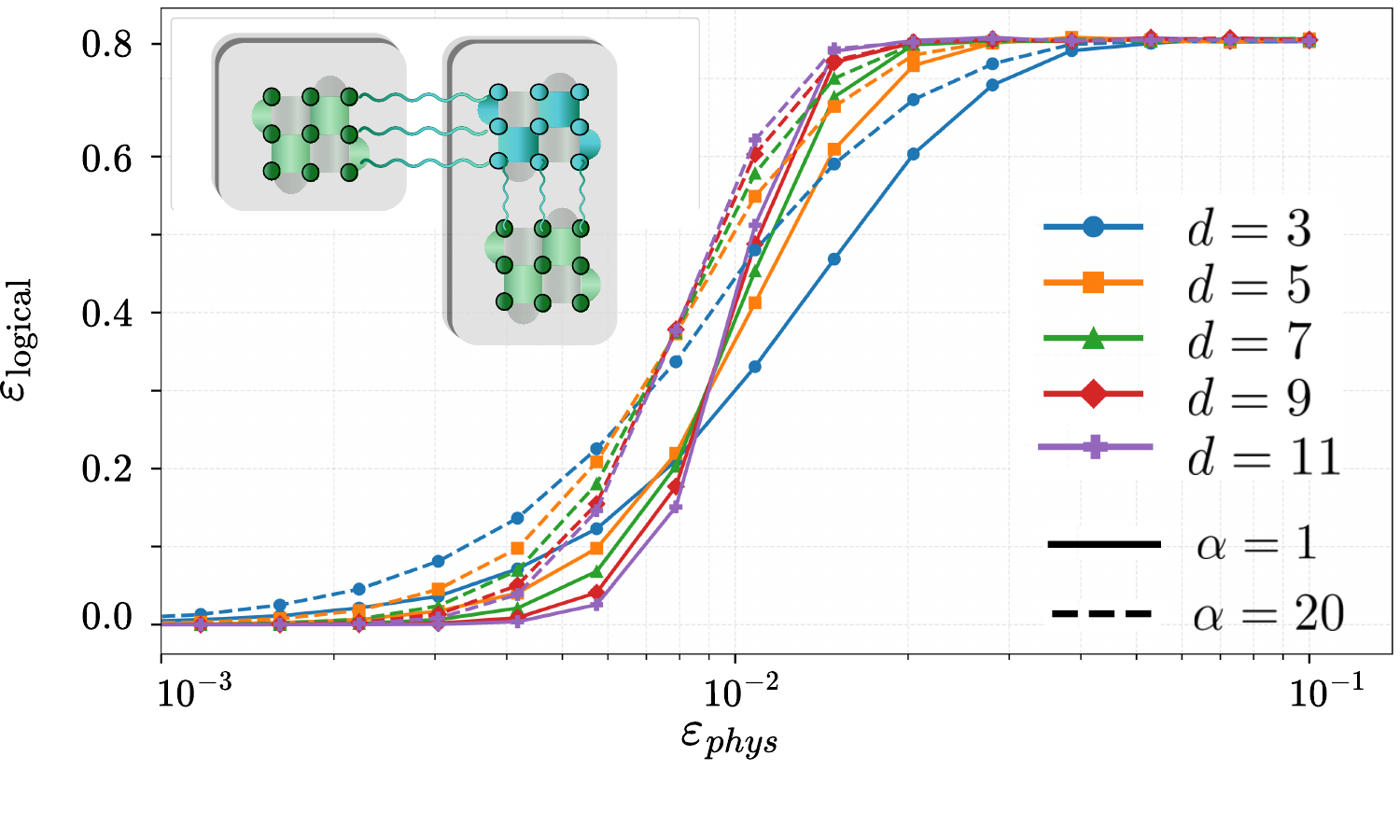}\\
        
        \includegraphics[width=1.00\linewidth]{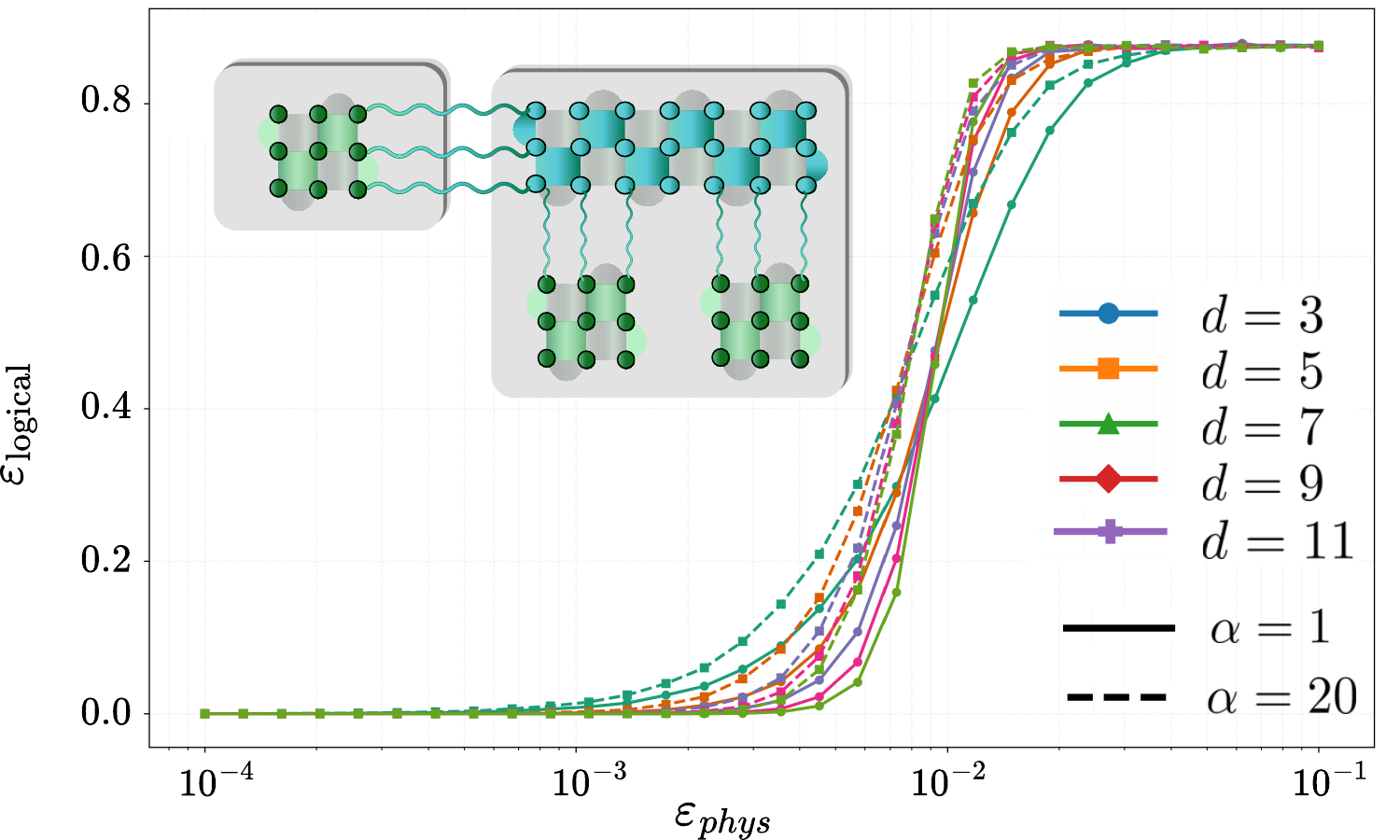}\\
        \vspace{0.5cm}
        \includegraphics[width=1.00\linewidth]{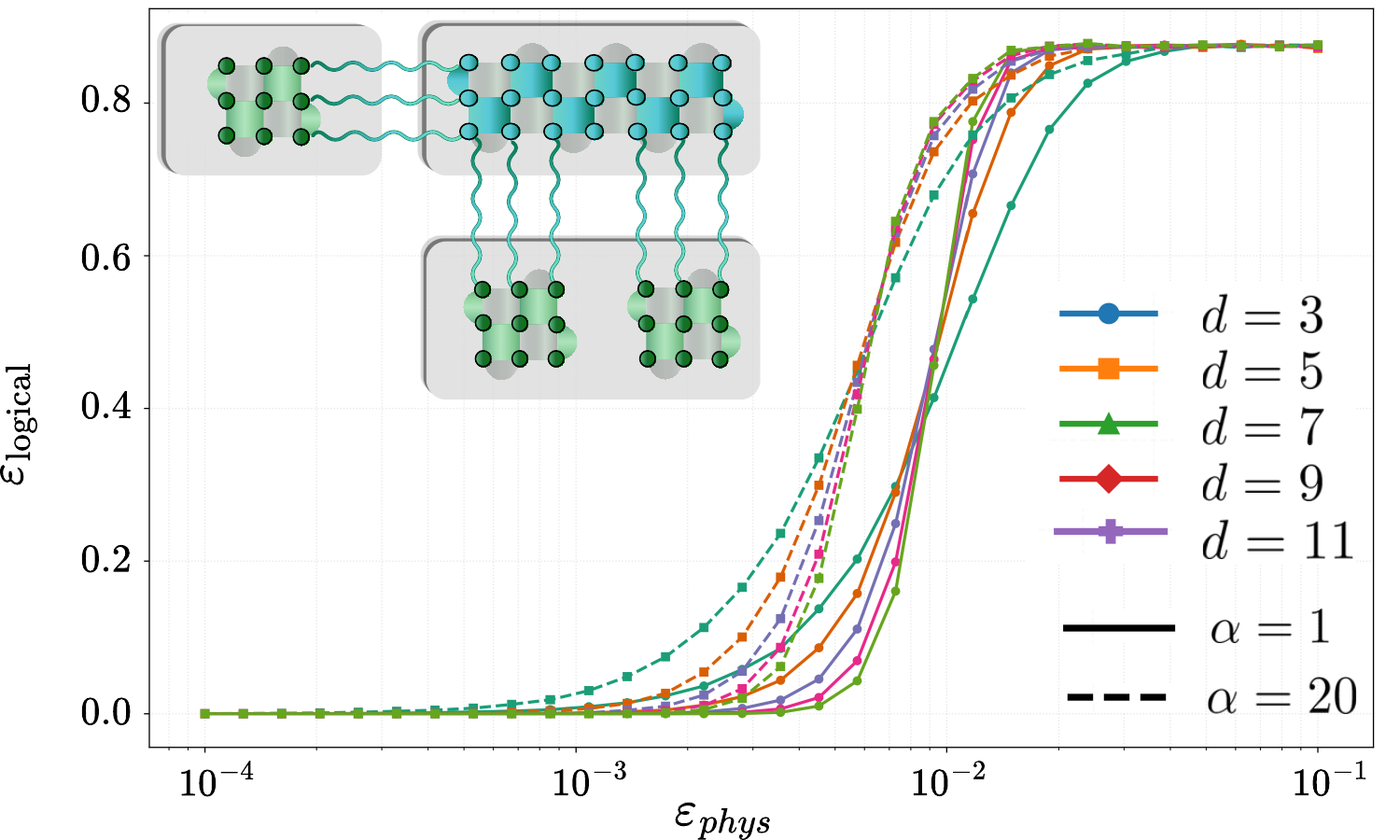}
        
    \caption{Threshold plots (logical error rate $\varepsilon_{\text{logical}}$ as a function of the physical error rate $\varepsilon_{\text{phys}}$) for three distributed gate configurations, each in comparison with the threshold plot for the monolithic case. From top to bottom: a single CNOT with one inter-QPU interface; two-target CNOT with one interface; and two-target CNOT with two interfaces. In all cases, threshold behavior is dominated by local gate noise, whereas inter-QPU gate errors have only a limited effect. As a result, the thresholds of the distributed configurations remain close to the corresponding monolithic thresholds at $\alpha=1$, even when the inter-QPU error rate is increased to $\alpha=20$.}
    \label{fig:three-column-composite}
\end{figure}

\newtheorem{theorem}{Theorem}
\newtheorem{remark}{Remark}

\section{Distributed Fault-tolerant GHZ State Preparation}

Having characterized the nonlocal logical CNOT primitive, which is key for distributed fault-tolerant multipartite entanglement generation, we now use it to construct distributed logical GHZ states and analyze their preparation resources. GHZ states are relevant to quantum sensing and networked quantum technologies \cite{Hyllus_2012,Kielinski_2024,Main_2025}, and provide key benchmarks for quantum processors \cite{moses2023racetrack,bao2024globalghz,cruz2019ghzw,chen2023ghz127,javadiabhari2025bigcats,chelluri2025}.

We describe the distributed architecture using graph theory. We consider a connected QPU network described by a graph $G=(V,E)$, where the vertices $V$ is the set of $n=|V|$ QPUs and the edges $E$ is the set of available inter-QPU links. Our goal is to prepare a distributed fault-tolerant logical GHZ state shared across the network. To this end, we develop an algorithm, inspired by Ref.~ \cite{chelluri2026}, and detailed in Algorithm~1 of Appendix~D. It first prepares local logical GHZ states within each QPU, and then iteratively fuses them across the network. After each merge, the qubit measured during the merging operation is reintroduced into the GHZ block via a local CNOT before proceeding to the next merge, until a single shared logical GHZ state is obtained. 

Implementing this construction through lattice surgery introduces resource requirements associated with both the logical encoding and the network architecture. We analyze three quantitiess: (i)~the number of QPUs equipped with ancilla patches, (ii)~the number of nonlocal Bell pairs consumed, and (iii)~the number of merging rounds. In the following, we examine how each of these quantities can be minimized. We first consider ancilla allocation.

A direct strategy is to place an ancilla patch in every QPU. However, we show that optimizing the ancilla allocation reduces to a minimum vertex cover problem (Theorem \ref{th:ancilla}). From now on, we refer to QPUs equipped with an ancilla patch as \textit{ancilla QPUs}. It is important to note that Algorithm 1 requires an additional ancilla patch in each ancilla QPU. 

\begin{theorem}\label{th:ancilla}
Under Algorithm~1, the minimum number of ancilla QPUs required to assemble a global logical GHZ state on a connected QPU network $G=(V,E)$ is
\[
a(G)=\min_{T\in\mathcal{T}(G)}\tau(T),
\]
where $\mathcal{T}(G)$ is the set of spanning trees of $G$ and  $\tau(T)$ denotes the minimum vertex-cover number of  $T$.
Moreover, for a fixed spanning tree $T$,
\[
a(T)=\tau(T)=\nu(T),
\]
where $\nu(T)$ is the maximum matching number of $T$.
\end{theorem}

The proof is given in Appendix~E. The minimum vertex cover problem is NP-hard in general graphs~\cite{karp1972_np}. To address this challenge, we introduce a heuristic approach for finding a low-cost vertex cover on the spanning tree. The main bottleneck is identifying the spanning tree that minimizes the size of the resulting vertex cover, as an exhaustive search over all possible spanning trees is computationally infeasible due to their exponential growth. Therefore, we consider trees generated using breadth-first search (BFS) and depth-first search (DFS)~\cite{newman2010networks} with those produced by greedy heuristic, \emph{Star-Search}, detailed in Algorithm 2 of Appendix~F. Star-Search constructs a tree by iteratively selecting high-degree center vertices to maximize the number of dominated nodes while reducing the required ancilla resources. 

Figure ~\ref{fig:vertex_cover} compares the average ancilla requirements of these methods, using 100 graph instances for each network size, with edge proability (0.3). The results show that spanning trees generated by DFS require, on average, the largest number of ancillas. This is because DFS tends to produce elongated, path-like trees, which generally have larger minimum vertex covers. In contrast, BFS typically produces more compact and highly branched trees and therefore requires fewer ancillas on average. Finally, our Star-Search algorithm consistently outperforms both DFS and BFS by producing spanning trees with smaller ancilla requirements across the considered graph instances.

\begin{figure}
    \centering
    \hspace*{-0.06\linewidth}
\includegraphics[width=1\linewidth]{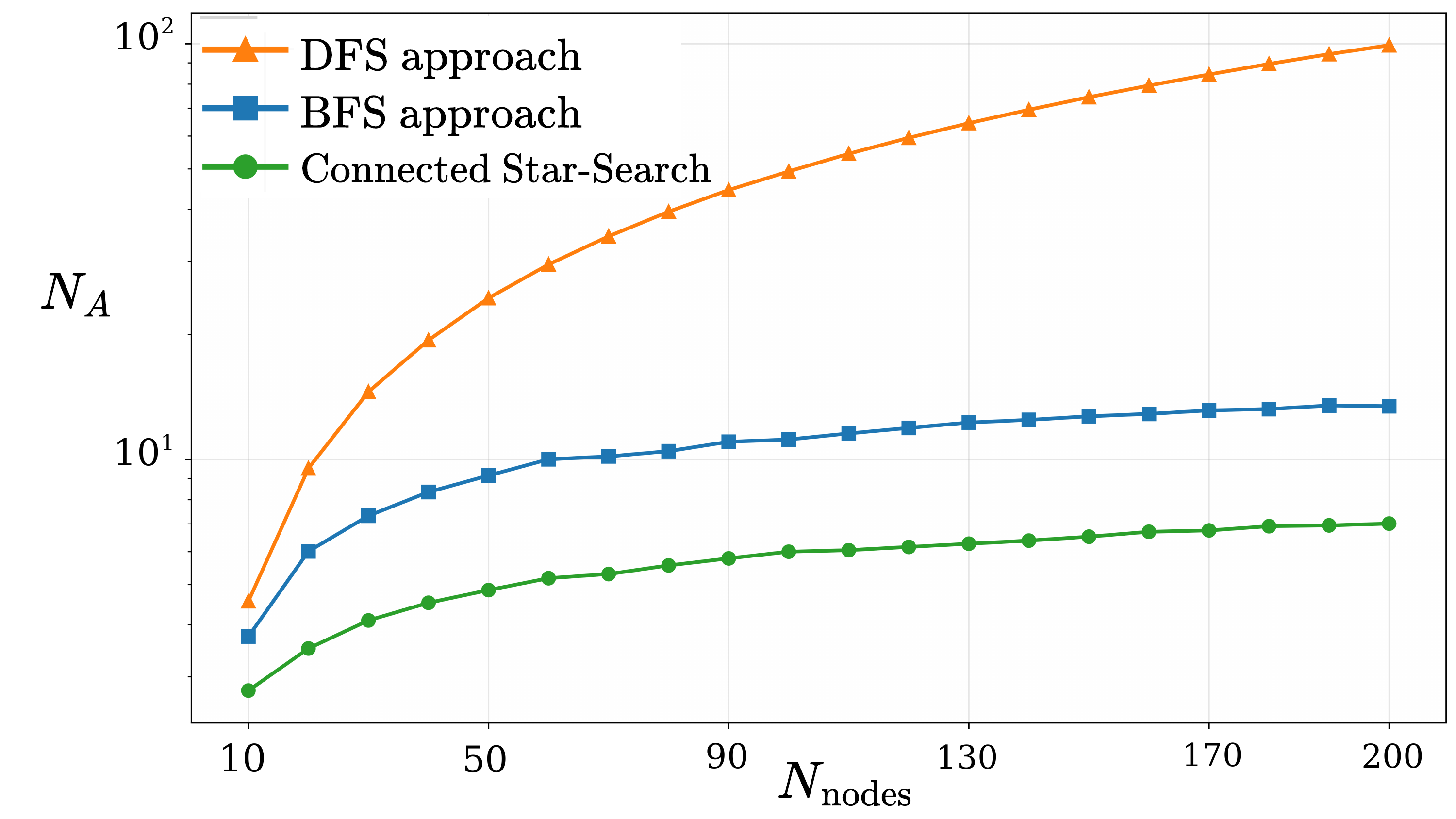}
    \caption{{Comparison of three methods for constructing a vertex cover on a spanning tree $T$. For each data point, we generate $100$ connected Erd\H{o}s--R\'enyi random graphs with edge probability $p=0.3$ and construct a spanning tree using depth-first search (DFS), breadth-first search (BFS), and an heuristic we developed called Star-Search (see Algorithm~2 in Appendix~F). The plot shows the average number of required ancilla QPUs, $N_A$, obtained by computing the exact minimum vertex cover of each resulting spanning tree.}}
    \label{fig:vertex_cover}
\end{figure}

We next consider inter-QPU entanglement consumption. Under Algorithm~1, each fusion of two GHZ blocks requires one nonlocal logical CNOT. These fusion operations are performed along the edges of a spanning tree connecting the QPUs, such that each fusion merges two previously disconnected GHZ blocks. A nonlocal logical CNOT between distance-$d$ rotated surface-code patches thereby consumes $d(2d-1)$ physical Bell pairs. Since this cost is the same for every fusion, minimizing the total Bell-pair consumption is equivalent to minimizing the number of inter-QPU fusion operations. Starting from $n$ local GHZ states and ending with a single global GHZ state requires exactly $n-1$ inter-QPU fusion steps, since each fusion combines two disconnected GHZ blocks and reduces the number of blocks by one. Therefore, the minimum physical Bell-pair consumption is $  d(2d-1)(n-1)$.

At fixed code distance, the consumption is determined by the number of QPUs. The preparation depth, however, depends on the network topology and ancilla allocation. A star network minimizes the ancilla requirement, since its minimum vertex cover consists solely of the central QPU. However, every fusion must involve this central ancilla, forcing the $n-1$ fusion operations to be performed sequentially and resulting in a depth of $D_{\mathrm{star}}=n-1$. In contrast, a linear network requires ancilla patches on every alternate QPU, corresponding to a minimum vertex cover of size $\lfloor n/2\rfloor$. This larger ancilla budget allows independent fusion operations to be carried out in parallel on disjoint portions of the network, yielding a depth of $D_{\mathrm{linear}}=\lceil\log_2 n\rceil$ under a balanced fusion schedule. More generally, $\lceil\log_2 n\rceil$ is the minimum depth attainable by any binary-fusion protocol and is achieved whenever the remaining GHZ blocks can be fused in disjoint pairs during every round. Thus, reducing the ancilla requirement from $O(n)$ to $O(1)$ increases the state-preparation depth from $O(\log n)$ to $O(n)$ rounds.

The preceding construction restores the logical patches consumed during fusion so that they remain part of the final distributed GHZ state. A relaxed objective is that at least one logical patch from each QPU be included. In that case, a patch consumed during an inter-QPU fusion need not be restored to its original local GHZ block, so the local restoration CNOTs can be omitted and the fusion schedule can be made more parallel. Under this relaxed requirement, each round may contain any set of inter-QPU fusions for which no QPU is used more than once. For a spanning tree $T$, this is equivalent to partitioning the edges of $T$, or equivalently to edge-coloring $T$, with each color representing one fusion round. Since every tree is bipartite, its edge-chromatic number equals its maximum degree $\Delta(T)$, and therefore the required number of rounds is
$ D_G=\Delta(T)$. 

Overall, these results show how network topology and ancilla allocation shape the resource trade-offs of distributed logical GHZ-state preparation. They provide a basis for choosing preparation strategies that balance qubit overhead and execution time within the connectivity constraints of a modular architecture.

\section{Conclusion}

In this work, we extended the assessment of noisy interconnects from logical-memory benchmarks to fault-tolerant logical operations between QPUs. Through circuit-level simulations of nonlocal logical CNOT gates implemented via lattice surgery in the rotated surface code, we quantified how local and interface errors affect logical reliability.  For the configurations and noise model studied, fault-tolerant behavior persists when inter-QPU gate error probabilities are twenty times those of local gates, with a limited reduction in the threshold relative to the corresponding monolithic implementations.

Building on this logical CNOT primitive, we developed a protocol for preparing distributed fault-tolerant logical GHZ states and analyzed its resource requirements for different QPU network topologies. We identified trade-offs between ancilla overhead and preparation time, and quantified the nonlocal Bell-pair consumption required by the fusion protocol. We further formulated ancilla allocation as the minimization of a vertex-cover number over spanning trees of the QPU network and introduced a polynomial-time heuristic for finding low-overhead solutions.

Together, these results connect the reliability of distributed logical operations with the architectural resources required to implement them. Extending the analysis to non-Clifford operations, communication-aware compilation, and tdecoding strategies that incorporate entanglement-generation latency will be important for assessing complete modular computations. By bringing logical-gate performance and network resource requirements into a common framework, this work informs the design of fault-tolerant quantum computers built from interconnected processors.

\vspace{0.8cm}
\begin{acknowledgments}
 This research has received funding from the European Innovation Council (EIC) Accelerator under Grant Agreement No.~101188682 for the SQOUT project and from the Hybrid HPC Quantum Initiative (HQI) project which is part of Plan France 2030. J.L. is a member of the Institut Universitaire de France.
\end{acknowledgments}

\bibliographystyle{apsrev4-2}
\bibliography{ref_edited}

\clearpage

\appendix

\section*{Appendix}

\section{Local Fault-Tolerant CNOT} 

\label{app:App1}Logical Clifford operations admit a complete description in the framework of stabilizer theory, acting by automorphisms of the Pauli group. Given a unitary $\hat{\mathcal U}$ its action is fully specified by its symplectic action on a generating set of Pauli operators, namely, by how it conjugates a Pauli basis. The characterization of the unitary is hence reduced to tracking its induced transformation under the stabilizer group and the logical Pauli operators.

Fault-tolerant implementation of unitaries can be performed via lattice surgery, a method of dynamical modification of stabilizers based on operations on surfaces with boundaries whose topology encodes logical degrees of freedom \cite{horsman2012latticesurgery, Litinski_2018,Litinski_2019, Chamberland_2022}. We consider (rotated) surface codes for encoding logical information: each code patch is a planar surface with distinct boundary types, determining whether strings corresponding to logical $\hat{Z}$ and $\hat{X}$ operators terminate there. 
Operations between surface codes admit a natural interpretation at the level of surface topology. 
The operation of gluing two surfaces along a common boundary is referred to as \textit{merging} and corresponds to the removal of a common boundary of two surface code patches and replacing it with a new set of stabilizers. The inverse operation which splits a given surface code patch, introducing new boundary conditions and correlated logical states, is referred to as \textit{splitting}. 
In contrast to transversal implementations, applying individual parallel physical gates, lattice surgery enacts logical operations though measurement-induced transformations of surfaces and their boundaries.

Along these lines, the implementation of a CNOT, given by the unitary $\hat{\mathcal{U}}=\vert 0\rangle\langle 0\vert \otimes\mathbb{I}+\vert 1\rangle\langle 1\vert\otimes \hat{X}$, admits a simple description in the stabilizer formalism. Given a control qubit $(c)$ and target qubit $(t)$ the operation reads: 
\begin{equation}
    X_c\to X_c X_t,\quad Z_c\to Z_c,\quad X_t\to X_t
\end{equation} 
The state of the logical qubit controlling the operation is encoded to a distance-$d$ rotated surface code patch, $\mathcal{C}_{c}[d]$, while the one of the target to a distinct patch $\mathcal{C}_{t}[d]$. The encoding of the logical control and target qubits is supplemented by an auxiliary surface code patch $\mathcal{C}_{aux}[d]$, prepared in a known logical state, acting as a mediator of the joint parity measurements comprising the CNOT operation.
The sequence of the two key stabilizer measurements is initiated with a joint $\hat{Z}^{(c)}_L \hat{Z}^{(a)}_L$ measurement between the control and ancilla, followed by a second joint $\hat{X}^{(a)}_L \hat{X}^{(t)}_L$ measurement between ancilla and target. These measurements are implemented through lattice-surgery operations between the code patches: merging of the control and auxiliary patches followed by splitting and merging of the auxiliary and target patches followed by splitting. As a final step, the ancilla is measured in the Z-basis and conditional Pauli corrections are applied depending on the measurement results.

\begin{figure}
\centering
\includegraphics[width=1.1\linewidth]{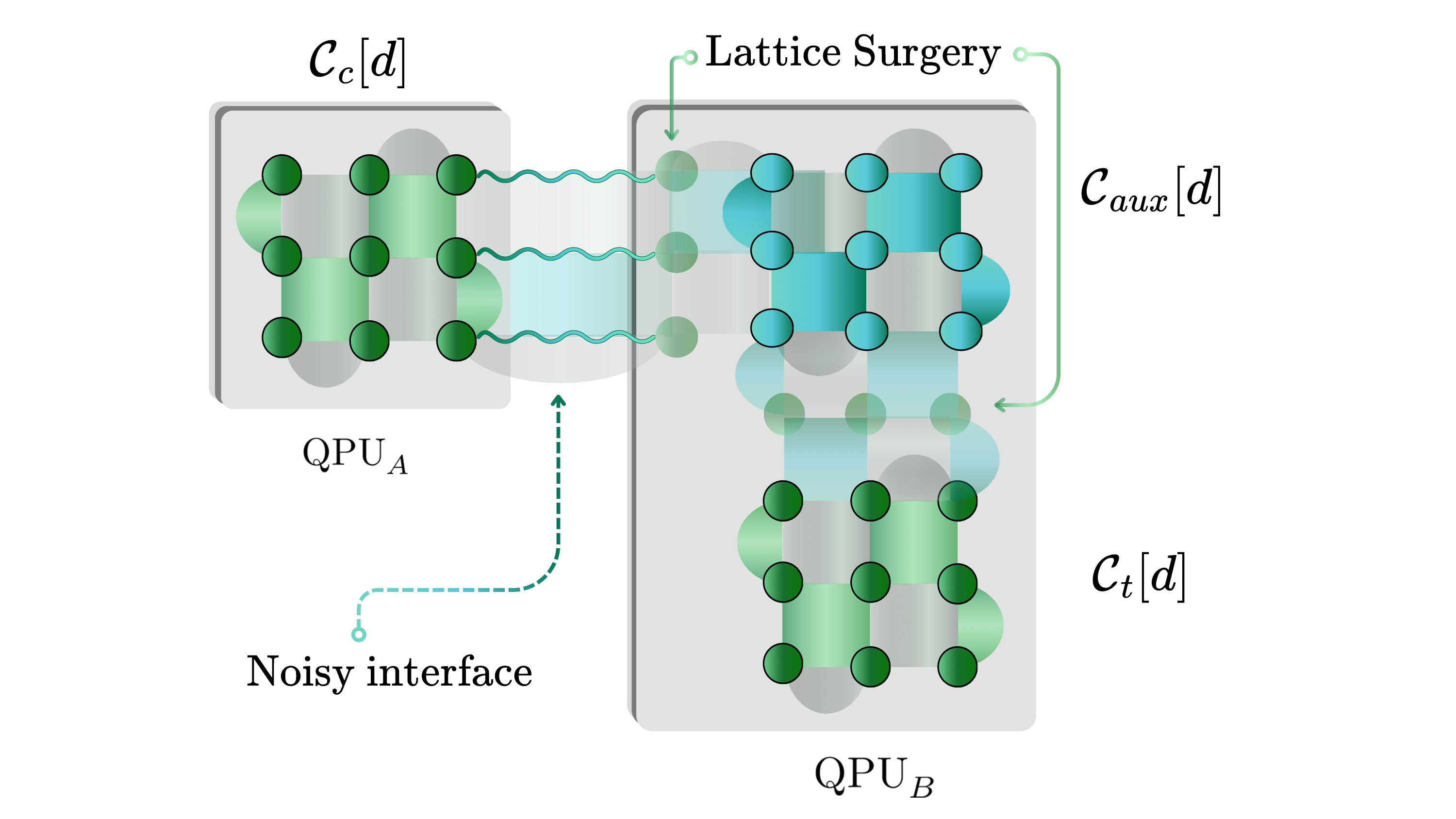} 
\caption{{Implementation of a distributed fault-tolerant CNOT. The control patch ($\mathcal{C}_{c}[d]$) resides in QPU$_A$, while the target ($\mathcal{C}_{t}[d]$) and auxiliary ($\mathcal{C}_{aux}[d]$) patches reside in QPU$_B$. Non-local lattice surgery between the control and auxiliary patches is mediated by a noisy interface, while the auxiliary-target operation remains local. We assume  that the physical data and syndrome qubits required for the control-auxiliary merge operation are all initialized in the second QPU. The light-colored regions indicate the non-local joint stabilizers, along with the corresponding qubits initialized for the lattice surgery operations.}}
\label{fig:non_local_cnot}
\end{figure}

\section{Distributed fault-tolerant multi-CNOT operations}
\label{app:App2} 
We consider a configuration consisting of multiple CNOT gates controlled by the same qubit. This sequence is realised fault-tolerantly via standard lattice surgery and is based on the nonlocal implementation of a single CNOT as explained in the main text and illustrated in Fig.~\ref{fig:non_local_cnot}. In our case, a rotated surface code patch encoding a logical control qubit interacts with a long ancilla which, in turn, interacts with a set of identical patches, each encoding a target qubit.

We are studying a set of architectural realizations ranging from the monolithic one, where all component patches are found in the same QPU, to the fully linear distributed setting, where the QPU hosting the control qubit is distinct from the rest of the  QPUs each of which is hosting an auxiliary and a target qubit. The QPUs are interconnected via noisy interfaces: the first interface is the one interconnecting the QPU hosting the control patch with the First QPU hosting an auxiliary and a target patch. The ones that follow, interconnect linearly the auxiliary patches of the distinct QPUs.

Given the code distance ($d$) of the surface code patches, the number of processors comprising the architecture ($N_p$) and the number of code patches encoding the target qubits of a controlled operation ($N_t$), we deduce that our configuration counts $2N_t+1$ code patches, and can have at maximum $2N_t$ interfaces, out of which the $N_p-1$ are nonlocal interfaces between ancillas and $2N_t-N_p+1$ are local interfaces between ancillas and target code patches. Focusing on the qubit counting, we consider a distance-$(d)$ rotated surface code patch, which includes $d^2$ physical data qubits and $d^2-1$ syndrome qubits (and hence equal number of stabilizers). After a merge of two such patches, the resulting configuration includes $2d^2+d$ data and $2d^2+d-1$ syndrome qubits. The total number of qubits required for the implementation of a local logical CNOT, as described in the first part of the appendix, is $6d^2+4d-1$.
In order to quantify the communication overhead in this distributed architecture, we focus on the ratio of nonlocal operations ($N_{nl}$) over the total number of local operations ($N_l$) of a distributed configuration, which we derive using the qubit and operation counting for the above architecture, taking into account the fact that stabilizers in the surface code consist of either four operations (bulk) or of two operations (boundary). The expressions read: 

\begin{equation} \mathcal{N}_{nl}(d,N_i) = d(2d-1)(N_p -1) \label{eq:nl} 
\end{equation} 

\begin{equation}
\begin{aligned}
\mathcal{N}_{l}(d,N_p,N_t)
={}& 4dN_t\Bigl((d+1)^2+(d^2-2)\Bigr)+\\- dN_p(2d-1) 
 &+ d(2d-1)^2+ 2d(d-1), && \label{eq:local} 
\end{aligned}
\end{equation}

\begin{figure}[!t]
   \centering
   \hspace{-6.5mm}
   \includegraphics[width=1.07\linewidth]{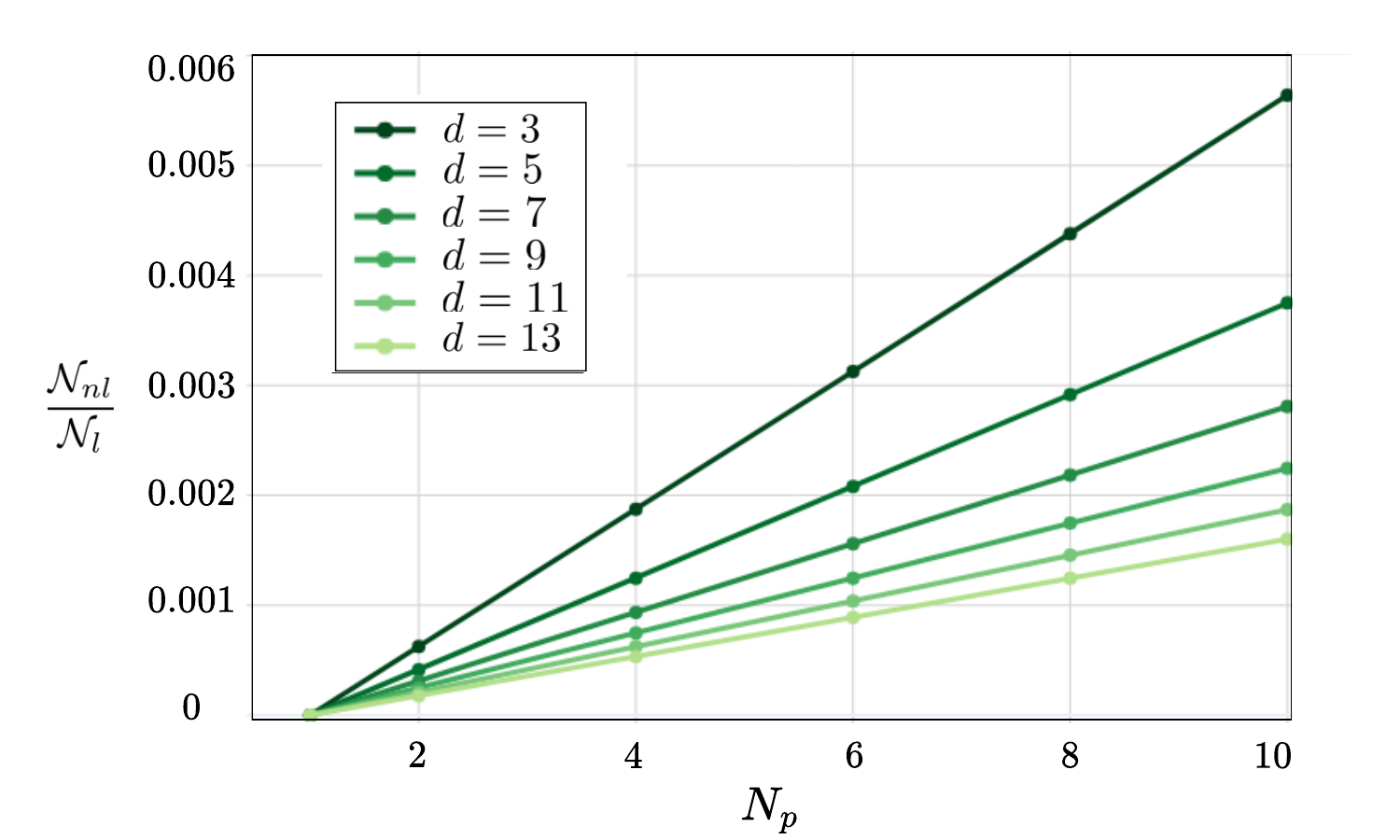}
  \caption{\justifying{Nonlocality ratio $\frac{\mathcal{N}_{nl}}{\mathcal{N}_{l}}$ as a function of the number of QPUs $N_{p}$ in a distributed configuration. The plotted example corresponds to the case of eight CNOTs implemented using distance-$d$ rotated surface code encoding.}}
  \label{fig:ratio}
   \end{figure}

We refer to the ratio $\mathcal{F}(d, N_p, N_t)=\frac{\mathcal{N}_{nl}}{\mathcal{N}_{l}}$ as the \textit{nonlocality factor}. From Fig.~\ref{fig:ratio}, which shows the non-locality factor as a function of the number of QPUs in a given configuration, we observe that it grows sub-linearly with the number of QPUs and decreases with the code distance. This behavior indicates that, even though distribution introduces additional inter-QPU communication, the relative contribution of nonlocal operations becomes less significant for larger code distances. Hence, distributed implementation of logical operations remain increasingly dominated by local fault-tolerant processing as the code distance scales.

\section{Choice of parameter $\alpha$}
\label{app:App3}
For a local CNOT gate, assume four independent elementary fault locations, corresponding to Pauli errors before and after the gate on each of the two qubits. If each fault location has error probability \(\epsilon\), the probability that no error occurs is
\begin{equation}
   (1-\epsilon)^4 . 
\end{equation}
Therefore, the probability that at least one error occurs is
\begin{equation}
\alpha_{\mathrm{loc}}
=
1-(1-\epsilon)^4 .
\end{equation}

Thus, expanding  for small \(\epsilon\), the leading order in the physical error probability is
\begin{equation}
\alpha_{\mathrm{loc}}
=
4\epsilon
+
O(\epsilon^2)
\simeq
4\epsilon .
\end{equation}
For the gate teleportation implementation of a nonlocal CNOT, we assume that  noisy physical operations are as shown in Fig.~\ref{fig:noisy-telegate}. The corresponding circuit contains $14$ fault locations. Since a local CNOT already accounts for $4$ of these, the gate-teleportation protocol introduces $10$ additional fault locations. Therefore, to first order in $\epsilon$, its additional error contribution is
\begin{equation}
\alpha_{\mathrm{TeleGate}}\simeq 10\epsilon .
\end{equation}

We also include the error due to the finite fidelity of the Bell pair. Let \(F_{\mathrm{Bell}}\) denote the fidelity of the distributed Bell pair with respect to the ideal state
\(\ket{\Phi^+}\). We model the effective accumulation of errors during Bell-pair preparation and distribution by introducing a coefficient $\alpha_{\mathrm{Bell}}$. In the absence of a hardware-specific entanglement-generation model, we set 
\begin{equation}
\alpha_{\mathrm{Bell}} \simeq 10 \epsilon.
\end{equation}
This corresponds to a Bell-pair infidelity that is an order of magnitude larger than the elementary intra-QPU error probability
\begin{equation}
F_{\mathrm{Bell}}
=
1-\alpha_{\mathrm{Bell}}
\simeq
1-10\epsilon
\end{equation}
Finally the total additional error strength of the nonlocal CNOT is,
\begin{equation}
\alpha
=
\alpha_{\mathrm{Bell}}
+
\alpha_{\mathrm{TeleGate}} \simeq 20 \epsilon.
\end{equation}

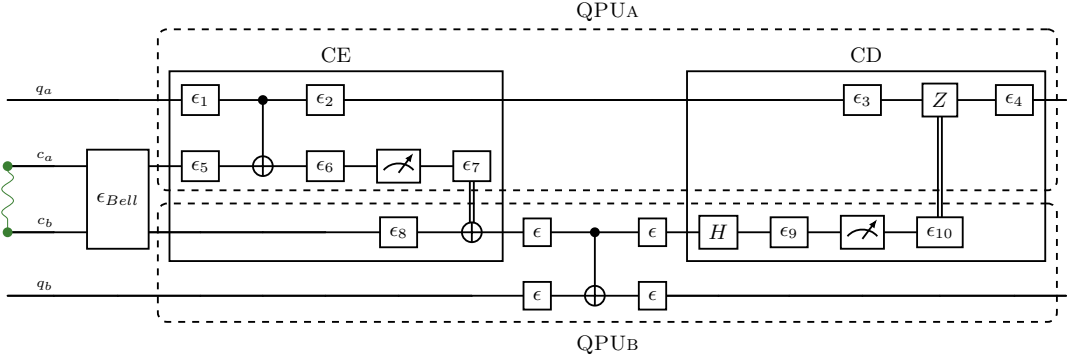
\begin{figure*}[!t]
\centering
\resizebox{0.8\textwidth}{!}{%
\begin{quantikz}[decoration=snake]
\wireoverride{n}
& \wire[l][1]["q_a"{above,pos=0.2}]{a}
& \qw
& \gate{\epsilon_1}
  \gategroup[2, steps=13,
  style={dashed, color=black, inner xsep=7pt, inner ysep=10pt, yshift=10pt, rounded corners},
  label style={label position=above, anchor=south, yshift=-6pt},
  background]{\textsc{QPUa}}
  \gategroup[3, steps=5,
  style={inner sep=2pt},
  background]{CE}
& \ctrl{1}
& \gate{\epsilon_2}
& \qw
& \qw
& \qw
& \qw
& \qw
& \qw
  \gategroup[3, steps=5,
  style={inner sep=2pt},
  background]{CD}
& \qw
& \gate{\epsilon_3}
& \gate{Z}
& \gate{\epsilon_4}
& \qw
\\
\phase[style={draw=OliveGreen, fill=OliveGreen}]{}\wireoverride{n}
& \wire[l][1]["c_a"{above,pos=0.2}]{a}
& \gate[2, disable auto height]{\epsilon_{Bell}}
& \gate{\epsilon_5}
& \targ{}
& \gate{\epsilon_6}
& \meter{}
& \gate{\epsilon_7}\wire[d][1]{c}
\\
\phase[style={draw=OliveGreen, fill=OliveGreen}]{}\wireoverride{n}
& \wire[l][1]["c_b"{above,pos=0.2}]{a}
&
& \qw
  \gategroup[2, steps=13,
  style={dashed, color=black, inner xsep=7pt, inner ysep=2pt, rounded corners},
  label style={label position=below, anchor=north, yshift=-8pt},
  background]{\textsc{QPUb}}
& \qw
& \qw
& \gate{\epsilon_8}
& \targ{}
& \gate{\epsilon}
& \ctrl{1}
& \gate{\epsilon}
& \gate{H}
& \gate{\epsilon_9}
& \meter{}
& \gate{\epsilon_{10}}\wire[u][2]{c}
\\
\wireoverride{n}
& \wire[l][1]["q_b"{above,pos=0.2}]{a}
& \qw
& \qw
& \qw
& \qw
& \qw
& \qw
& \gate{\epsilon}
& \targ{}
& \gate{\epsilon}
& \qw
& \qw
& \qw
& \qw
& \qw
& \qw
\arrow[from=2-1, to=3-1, OliveGreen, decorate, -] {}
\end{quantikz}}
\caption{{Noisy gate teleportation protocol, {structured via Cat-Entangler (CE) and Cat-Disentangler (CD) primitives} with {Pauli errors from input Bell-pair, gates and measurements.}}}
\label{fig:noisy-telegate}
\end{figure*}

\section{Algorithm for distributed GHZ construction}\label{app:App4}

Algorithm~\ref{alg:ft_ghz_network} summarizes the procedure used to generate a distributed logical GHZ state over a connected QPU network \(G=(V,E)\), where \(V\) denotes the set of QPUs and \(E\) the available inter-QPU links. We extend the GHZ-fusion protocol of Ref.~\cite{chelluri2026} to logical qubits encoded in rotated surface-code patches hosted on distinct QPUs. The protocol begins by preparing a local logical GHZ state within each QPU, and then successively fuses these local GHZ blocks along a connected spanning set of inter-QPU links until a single logical GHZ state shared across all QPUs is obtained.

\begin{algorithm}[H]
\caption{Fault-tolerant GHZ assembly on a QPU network}
\label{alg:ft_ghz_network}
\begin{algorithmic}[1]
\Require Connected QPU network $G=(V,E)$.
\Ensure A global logical GHZ state shared by all QPUs.
\State For each $v \in V$, prepare a local logical GHZ state by
       initializing the RSC patch and performing a smooth split.
\State Choose a connected edge set $F \subseteq E$ that spans all QPUs.
\State Merge two GHZ states using the edges in $F$,
       as described in Protocol 3 in Ref.~\cite{chelluri2026}.
\State Re-integrate the qubit measured in the previous step into the
       GHZ block by performing a CNOT with the local GHZ state.
\State Repeat the two preceding steps until all GHZ states have
       been merged into a single logical GHZ state shared by all QPUs.
\end{algorithmic}
\end{algorithm}

\section{Optimal ancilla allocation } \label{app:App5}
\noindent The Proof of Theorem \ref{th:ancilla} is provided below.

\begin{proof}
Algorithm~\ref{alg:ft_ghz_network} first prepares a local logical GHZ state in each QPU and then merges these states by applying nonlocal logical CNOTs along a connected set of inter-QPU edges
$F\subseteq E$. Since the purpose of $F$ is only to connect all QPUs, any cycle edge in $F$ is redundant: deleting it leaves the set of QPUs still connected and does not obstruct the remaining GHZ-merging steps. Repeating this deletion process reduces $F$ to a spanning tree $T\subseteq G$. Hence any feasible protocol may be reduced to one whose fusion edges form a spanning tree.

Fix such a spanning tree $T$. Let $A\subseteq V$ be the set of QPUs that host ancilla patches. In every merge step across an edge $(u,v)\in E(T)$, the nonlocal logical CNOT must have its target in an ancilla-enabled QPU. Therefore at least one endpoint of every used edge
must belong to $A$, i.e.,
\[
\forall\, (u,v)\in E(T),\qquad u\in A \text{ or } v\in A.
\]
This is exactly the definition of a vertex cover of $T$. Conversely, if $A$ is a vertex cover of $T$, then every edge of $T$ has at least one ancilla-enabled endpoint, so each merge step required by Algorithm~\ref{alg:ft_ghz_network} is feasible. Thus, for a fixed tree $T$, feasible ancilla assignments are precisely the vertex covers of $T$, and the minimum ancilla cost on $T$ is $a(T)=\tau(T).$

Finally, since one is free to choose the spanning tree used for merging, the globally optimal ancilla cost is obtained by minimizing over all spanning trees of $G$, identified as $\mathcal{T}(G)$:
\[
a(G)=\min_{T\in \mathcal{T}(G)}\tau(T).
\]

Because every tree is bipartite, K\"onig's theorem applies and gives $\tau(T)=\nu(T)$ for every fixed spanning tree $T$. Hence \[
a(T)=\tau(T)=\nu(T),
\]
which identifies the optimal ancilla count on a fixed merge tree with both its minimum vertex cover and its maximum matching.
\end{proof}

\section{Numerical search for low-cover spanning trees}\label{app:App6}

By Theorem~\ref{th:ancilla}, ancilla minimization reduces to finding a spanning tree of the QPU network with small vertex-cover number. This tree-selection step is the main bottleneck: different spanning trees of the same network can lead to substantially different ancilla costs.

For any fixed tree $T$, the ancilla cost can be computed exactly and efficiently by dynamic programming. Root $T$ at a vertex $r$, and for each node $u$ define $I_u$ as the minimum cover size of the subtree rooted at $u$ when $u$ is included in the cover, and $O_u$ as the minimum cover size when $u$ is excluded. If ${\rm ch}(u)$ denotes the children of $u$, then
\begin{align}
I_u &= 1 + \sum_{v\in {\rm ch}(u)} \min\{I_v,O_v\}, \\
O_u &= \sum_{v\in {\rm ch}(u)} I_v,
\end{align}
and the exact cover number is
\begin{equation}
\tau(T)=\min\{I_r,O_r\}.
\label{eq:treevcdp}
\end{equation}

Thus, once a spanning tree is fixed, computing its ancilla cost is straightforward; the nontrivial task is to identify a spanning tree with low vertex-cover number. In the numerics, we compare three ways to choose the spanning tree: the BFS tree and the DFS tree from a chosen root, and a greedy connected-star heuristic designed specifically for ancilla reduction. The heuristic aims to build a tree with many leaves, motivated by the fact that in any tree the internal vertices form a vertex cover, so a larger number of leaves is generally associated with a smaller ancilla requirement.

\begin{figure}[h]
\begin{algorithm}[H]
\caption{Star-Search algorithm for ancilla placement}
\label{alg:connectedstar}

\begin{algorithmic}[1]
\Require A connected QPU network $G=(V,E)$
\Ensure A spanning tree $T\subseteq G$ and an ancilla set $A\subseteq V$
\State Compute the degree of every vertex in $G$
\State Choose a highest-degree vertex $c_0$ as the first center
\State Initialize the center set $S$ so that it contains only $c_0$
\State Initialize the dominated set $D$ so that it contains $c_0$ and all vertices adjacent to              $c_0$
\State Start with no tree edges selected
\State Repeat the next five steps (7--11) until every vertex of $G$ belongs to $D$
\State Among the vertices not already in $S$ but adjacent to at least one vertex in $S$, choose a           vertex $v^{\star}$ that would newly dominate the largest number of vertices; if there is            a tie, choose one of largest degree
\State Connect $v^{\star}$ to a neighboring vertex already in $S$; if there is a choice, use a              neighboring center of largest degree
\State Add $v^{\star}$ to the center set $S$
\State Add the new edge to the tree-edge set
\State Add $v^{\star}$ and all of its neighbors to the dominated set $D$
\State Initialize the tree $T$ with all vertices of $G$ and with the selected tree edges
\State For each remaining vertex $u$ not already in $S$, attach $u$ to a neighboring vertex in $S$          of largest degree
\State Compute $\tau(T)$ using Eq.~\eqref{eq:treevcdp} and recover a corresponding minimum vertex           cover $A$ by traceback
\State \Return $(T,A)$
\end{algorithmic}
\end{algorithm}
\end{figure}

Algorithm~\ref{alg:connectedstar} should be read as follows. The set $S$ contains the current center vertices, and the set $D$ contains the vertices already dominated by $S$. In each pass through the repeated block, the algorithm chooses one new center $v^{\star}$, connects it to the existing center set, adds the corresponding edge to the tree, and enlarges the dominated set. The loop stops once every vertex belongs to $D$. At that point the center set is connected, but some non-center vertices may not yet appear explicitly as leaves of the tree. The next step therefore attaches each remaining vertex directly to a neighboring center. The resulting spanning tree is then evaluated exactly by computing a minimum vertex cover on that tree, which yields the ancilla set. The heuristic is designed to produce spanning trees with many leaves, which are expected to have relatively small ancilla cost.

\begin{theorem}
\label{thm:connectedstarcomplexity}
Let $G=(V,E)$ be a connected graph with $n=|V|$ and $m=|E|$, represented by adjacency lists. Algorithm~\ref{alg:connectedstar} constructs a spanning tree $T\subseteq G$ and computes an optimal ancilla set $A$ for that tree in $O(nm)$ time and $O(n+m)$ memory.
\end{theorem}

\begin{proof}
The degree computation costs $O(n+m)$. Each pass through the while-loop adds one new center to $S$, so the loop executes at most $n-1$ times. In one pass, the algorithm scans adjacency lists to identify the admissible candidates and evaluate their gains, which costs $O(n+m)$ time. Therefore the full center-growth stage costs $O(n(n+m))$.

Because $G$ is connected, $m\ge n-1$, so $O(n(n+m))=O(nm)$. After the center set has been chosen, the remaining vertices are attached to the centers by scanning their neighborhoods once, which costs $O(n+m)$. The exact minimum vertex-cover size of the final tree is then computed using Eq.~\eqref{eq:treevcdp}, which costs $O(n)$ because the tree has $n-1$ edges. Recovering a corresponding minimum vertex cover by traceback through the dynamic-programming table also costs $O(n)$. Hence the total running time is $O(nm)$.

For memory, the adjacency lists use $O(n+m)$ space, and all additional sets, parent arrays, and dynamic-programming tables use only $O(n)$ space. Hence the total memory requirement is $O(n+m)$.
\end{proof}

In summary, Theorem~\ref{th:ancilla} reduces the ancilla-allocation problem to the search for a spanning tree with small vertex-cover number.  BFS and DFS provide simple traversal-based baselines and Algorithm~\ref{alg:connectedstar} is a purpose-built heuristic for constructing trees with low ancilla cost.

\end{document}